\documentclass[a4paper,11pt]{article}

\usepackage{fullpage}

\usepackage{times}
\usepackage{soul}
\usepackage{url}
\usepackage[hidelinks]{hyperref}
\usepackage[utf8]{inputenc}
\usepackage[small]{caption}
\usepackage{graphicx}
\usepackage{amsmath}
\usepackage{booktabs}
\usepackage{algorithm,algpseudocode}
\usepackage{amsthm}
\usepackage{amssymb}

\usepackage{color}
\usepackage{xspace}
\usepackage{tikz}

\newtheorem{theorem}{Theorem}[section]

\newtheorem{lemma}[theorem]{Lemma}

\theoremstyle{definition}
\newtheorem{definition}[theorem]{Definition}

\usepackage{natbib}

\usepackage{authblk}

\newcommand{\OPT}{{\mathsf{OPT}\xspace}}
\newcommand{\SW}{{\mathsf{SW}\xspace}}
\newcommand{\cM}{\mathcal{M}}
\DeclareMathOperator*{\argmax}{arg\,max}

\newcommand{\specialcell}[2][c]{%
  \begin{tabular}[#1]{@{}c@{}}#2\end{tabular}}

\allowdisplaybreaks

\begin{document}

\title{\bf The Price of Fairness for Indivisible Goods\thanks{
A preliminary version of the paper appeared in Proceedings
of the 28th International Joint Conference on Artificial Intelligence \citep{BeiLuMa19}.
The third author is now at Google Research.
Most of the work was done while the fourth author was at the University of Oxford.
}}

\author[1]{Xiaohui Bei}
\author[1]{Xinhang Lu}
\author[2]{Pasin Manurangsi}
\author[3]{Warut Suksompong}

\affil[1]{School of Physical and Mathematical Sciences, Nanyang Technological University}
\affil[2]{Department of Electrical Engineering and Computer Sciences, UC Berkeley}
\affil[3]{School of Computing, National University of Singapore}

\date{}

\maketitle

\begin{abstract}
We investigate the efficiency of fair allocations of indivisible goods using the well-studied \emph{price of fairness} concept. Previous work has focused on classical fairness notions such as envy-freeness, proportionality, and equitability. However, these notions cannot always be satisfied for indivisible goods, leading to certain instances being ignored in the analysis. In this paper, we focus instead on notions with guaranteed existence, including envy-freeness up to one good (EF1), balancedness, maximum Nash welfare (MNW), and leximin. We also introduce the concept of \emph{strong price of fairness}, which captures the efficiency loss in the worst fair allocation as opposed to that in the best fair allocation as in the price of fairness. We mostly provide tight or asymptotically tight bounds on the worst-case efficiency loss for allocations satisfying these notions, for both the price of fairness and the strong price of fairness.
\end{abstract}

\section{Introduction}

The allocation of scarce resources among interested agents is a problem that arises frequently and plays a major role in our society. We often want to ensure that the allocation that we select is \emph{fair} to the agents---the literature of \emph{fair division}, which dates back to the design of cake-cutting algorithms over half a century ago \citep{Steinhaus48,DubinsSp61}, provides several ways of defining what fair means. 
For example, an allocation is \emph{envy-free} if it does not generate envy between any pair of agents, \emph{proportional} if it gives every agent $1/n$ of the agent's utility for the whole set of resources (here $n$ denotes the number of agents), and \emph{equitable} if every agent receives the same utility. 
An issue orthogonal to fairness is \emph{efficiency}, or \emph{social welfare}, which refers to the total happiness of the agents. A fundamental question is therefore how much efficiency we might lose if we want our allocation to be fair.

This question was first addressed independently by \citet{BertsimasFaTr11} and \citet{CaragiannisKaKa12}, who introduced the \emph{price of fairness} concept to capture the efficiency loss due to fairness constraints. In particular, for any fairness notion and any given resource allocation instance with additive utilities, Caragiannis et al.~defined the price of fairness of the instance to be the ratio between the maximum social welfare over all allocations and the maximum social welfare over allocations that are fair according to the notion. The overall price of fairness for this notion is then defined as the largest price of fairness across all instances. Caragiannis et al.~considered the three aforementioned fairness notions and presented a series of results on the price of fairness with respect to these notions; they assumed that the agents have additive utilities and each agent has utility $1$ for the entire set of resources. As an example, they showed that for the allocation of indivisible goods among $n$ agents, the price of proportionality is $n-1+1/n$, meaning that the efficiency of the best proportional allocation can be a linear factor away from that of the best allocation overall.

Caragiannis et al.'s work sheds light on the trade-off between efficiency and fairness in the allocation of both divisible and indivisible resources. However, a significant limitation of their study is that while an allocation satisfying each of the three fairness notions always exists when goods are divisible, this is not the case for indivisible goods. Indeed, none of the notions can be satisfied in the simple instance with (at least) two agents and a single good to be allocated. Caragiannis et al.~circumvented this issue by simply ignoring instances in which the fairness notion in question cannot be satisfied. As a result, their price of fairness analysis, which is meant to capture the worst-case efficiency loss, fails to cover certain scenarios that may arise in practice.\footnote{From the above example, one may think that such scenarios are rare exceptions. However, for envy-freeness, these scenarios are in fact common if the number of goods is not too large compared to the number of agents \citep{DickersonGoKa14,ManurangsiSu20}.} In addition, the fact that certain instances are not taken into account in the price of fairness have seemingly contradictory consequences. For example, since envy-free allocations are always proportional when utilities are additive, it may appear at first glance that the price of envy-freeness must be at least as high as the price of proportionality. This is not necessarily the case, however, because there are instances that admit proportional but no envy-free allocations.\footnote{Indeed, the instance that Caragiannis et al.~used to show that the price of proportionality is at least $n-1+1/n$ admits no envy-free allocation. Thus, it is still possible that the price of envy-freeness is lower than the price of proportionality for indivisible goods.}

To address these limitations, in this paper we study the price of fairness for indivisible goods with respect to fairness notions that can be satisfied in every instance. Among other notions, we consider envy-freeness up to one good (EF1), balancedness, maximum Nash welfare (MNW), maximum egalitarian welfare (MEW), and leximin.\footnote{See Section~\ref{sec:prelim} for the formal definitions of these notions.} 
For example, in an EF1 allocation, an agent may envy another agent, but in such a case there must exist a good in the second agent's bundle such that the envy disappears upon removing the good.
An MNW allocation maximizes the product of the utilities that the agents receive, while an MEW allocation maximizes the minimum among these utilities.
In addition to deriving bounds on the price of fairness for these notions, we also introduce the concept of \emph{strong price of fairness}, which captures the efficiency loss in the worst fair allocation as opposed to that in the best fair allocation. 
The strong price of fairness is relevant in settings where one is guaranteed an allocation satisfying some fairness notion but has no control over the particular allocation---for instance, we may be participating in an allocation exercise using an algorithm that guarantees EF1 or MNW, and wonder whether that fairness guarantee comes with any assurance on the social welfare.
Indeed, certain fair division algorithms such as the \emph{envy cycle elimination algorithm} \citep{LiptonMaMo04} may output EF1 allocations with low efficiency.\footnote{See the example in Theorem~\ref{thm:EF1-EFX-PoA}.}
The relationship between the price of fairness and the strong price of fairness is akin to that between the price of stability and the price of anarchy for equilibria. While the strong price of fairness is too demanding to yield any nontrivial guarantee for some fairness notions, as we will see, it does provide meaningful guarantees for other notions. 

\subsection{Our Results}

\begin{table*}[t]
\begin{center}
    \begin{tabular}{ | c || c | c || c | c | }
    \hline
     \textbf{Property $P$} & \multicolumn{2}{ c|| }{\textbf{Price of $P$}} & \multicolumn{2}{ c| }{\textbf{Strong price of $P$}} \\ \hline \hline
      & General $n$ & $n=2$ & General $n$ & $n=2$  \\ \hline
		Envy-freeness up to one good (EF1) & \specialcell{LB: $\Omega(\sqrt{n})$ \\ UB: $O(n)$} & \specialcell{LB: $8/7$ \\ UB: $2/\sqrt{3}$} & $\infty$ & $\infty$   \\ \hline
		Envy-freeness up to any good (EFX) & $-$ & $3/2$ & $-$ & $\infty$   \\ \hline
		Round-robin (RR) & $n$ & $2$ & $n^2$ & $4$   \\ \hline
		Balancedness (BAL) & $\Theta(\sqrt{n})$ & $4/3$ & $\infty$ & $\infty$   \\ \hline
		Maximum Nash welfare (MNW) & $\Theta(n)$ & \specialcell{LB: $27/23$ \\ UB: $5/4$}  & $\Theta(n)$ & \specialcell{LB: $27/23$ \\ UB: $5/4$}   \\ \hline
		Maximum egalitarian welfare (MEW) & $\Theta(n)$ & $3/2$ & $\infty$ for $n\geq 3$ & $3/2$   \\ \hline
		Leximin (LEX) & $\Theta(n)$ & $3/2$ & $\Theta(n)$ & $3/2$   \\ \hline
		Pareto optimality (PO) & $1$ & $1$ & $\Theta(n^2)$ & $3$   \\ \hline
    \end{tabular}
    \caption{Summary of our results. LB denotes lower bound and UB denotes upper bound. We do not consider the (strong) price of EFX for general $n$ because it is not known whether an EFX allocation always exists for $n > 3$. If we allow dependence on the number of goods $m$, we have an upper bound of $O(\sqrt{n}\log(mn))$ on the price of EF1.}
    \label{table:summary}
\end{center}
\end{table*}

The majority of our results can be found in Table~\ref{table:summary}; we highlight a subset of these next. For the price of EF1, we provide a lower bound of $\Omega(\sqrt{n})$ and an upper bound of $O(n)$. We then show that two common methods for obtaining an EF1 allocation---the round-robin algorithm and MNW---have a price of fairness of linear order (for round-robin the price is exactly $n$), implying that these methods cannot be used to improve the upper bound for EF1. 
On the other hand, if we allow dependence on the number of goods $m$, the price of round-robin, and therefore the price of EF1, is $O(\sqrt{n}\log(mn))$---this means that the $\Omega(\sqrt{n})$ lower bound is almost tight unless the number of goods is huge compared to the number of agents. 
Our result illustrates a clear difference between EF1 and envy-freeness, as the price of the latter is $\Theta(n)$ \citep{CaragiannisKaKa12}.
For MNW, MEW, and leximin, we prove an asymptotically tight bound of $\Theta(n)$ on the price of fairness. Moreover, with the exception of EF1 and MNW, we establish exactly tight bounds in the case of two agents for all fairness notions---in particular, the price of EF1 is between $1.14$ and $1.16$, implying that there exists an EF1 allocation whose welfare is close to the optimal welfare.

Our results point to round-robin as a promising allocation method: besides producing an EF1 allocation with high welfare, it is extremely simple and intuitive, and an allocation that it produces is always balanced.\footnote{Moreover, a round-robin allocation is likely to be envy-free and proportional as long as the number of goods is sufficiently larger than the number of agents \citep{ManurangsiSu20-2}.}
Most of our upper bounds naturally give rise to polynomial-time algorithms for computing an allocation satisfying the fairness notion with the guaranteed welfare.
However, there are two notable exceptions:\footnote{In addition to these exceptions, MNW, MEW, and leximin allocations are hard to compute regardless of price of fairness considerations (see, e.g., \citep[footnote~7]{PlautRo20}).} (i) the proof of Theorem~\ref{thm:EFX-PoS-2} requires an agent to partition the goods into two bundles such that her utilities for the bundles are as equal as possible, an NP-hard problem; (ii) the upper bound in Theorem~\ref{thm:round-robin-PoA}, which relies on Lemma~\ref{lem:roundrobin-welfare}, is based on a randomized approach and does not indicate how a desirable round-robin ordering can be efficiently chosen.

On the strong price of fairness front, we show via a simple instance that the strong price of EF1 and balancedness are infinite, meaning that there are arbitrarily bad EF1 and balanced allocations. Nevertheless, a round-robin allocation, which satisfies these two properties, always has welfare within a factor $n^2$ of the optimal allocation, and this factor is exactly tight. For MNW and leximin, the strong price of fairness, like the price of fairness, is of linear order---hence, these two notions provide a better worst-case guarantee than the round-robin algorithm. However, while the price of MEW is also $\Theta(n)$, the strong price of MEW is infinite for $n\geq 3$ (and $3/2$ for $n=2$), meaning that a MEW allocation does not provide any welfare guarantee when there are at least three agents. Finally, we consider Pareto optimality, for which the price of fairness is trivially $1$, and show that the strong price of Pareto optimality is $\Theta(n^2)$. 
This demonstrates that an allocation that is optimal in the Pareto sense may be quite far from optimal with respect to social welfare.

%Finally, we relate prices of fairness in the indivisible and the divisible settings by showing that, for the same number of agents, the price of EF1 (for indivisible goods) is no smaller than the price of EF (for divisible goods).

\subsection{Related Work}

As we mentioned earlier, the price of fairness was introduced independently by \citet{BertsimasFaTr11} and \citet{CaragiannisKaKa12}. Bertsimas et al.~studied the concept for divisible goods with respect to fairness notions such as proportional fairness and max-min fairness; in particular, their results on proportional fairness imply that the price of envy-freeness and the price of MNW for divisible goods are both $\Theta(\sqrt{n})$.\footnote{Interestingly, this stands in contrast to our result that the price of MNW for indivisible goods is $\Theta(n)$.} Caragiannis et al. presented a number of bounds for both goods and \emph{chores} (i.e., items that yield negative utility), both when these items are divisible and indivisible.  The price of fairness has subsequently been examined in several other settings, including for \emph{contiguous} allocations of divisible goods \citep{AumannDo15}, indivisible goods \citep{Suksompong19}, and divisible chores \citep{HeydrichVa15}, as well as in the context of machine scheduling \citep{BiloFaFl16} and budget division \citep{MichorzewskiPeSk20}.

Typically, the price of fairness study focuses on quantifying the efficiency loss solely in terms of the number of agents. A notable exception to this is the work of \citet{Kurz14}, who remarked that certain constructions used to establish worst-case bounds for indivisible goods require a large number of goods. As a result, Kurz investigated the dependence of the price of fairness on both the number of agents and the number of goods, and, as we do for the price of round-robin, found that the price indeed improves significantly if we limit the number of goods.

Since envy-freeness and proportionality cannot always be satisfied even in the simplest setting with two agents and one good, a large number of recent papers have focused on relaxations of these notions, which include EF1, EFX, maximin share (MMS), and pairwise maximin share (PMMS) \citep{AmanatidisBiMa18,BiswasBa18,GhodsiHaSe18,CaragiannisKuMo16,OhPrSu19,KyropoulouSuVo20,PlautRo20}.\footnote{See \citep{CaragiannisKuMo16} for the definitions of MMS and PMMS.}
It is known that MMS allocations do not necessarily exist, while the existence question is open for PMMS \citep{KurokawaPrWa18,CaragiannisKuMo16}.
We refer to \citep{Markakis17} and \citep{CaragiannisKuMo16} for further discussion of work on these notions.

After the publication of the initial version of our work, \citet{BarmanBhSh20} devised an algorithm that produces an allocation with social welfare within $O(\sqrt{n})$ of the optimum; together with our result, this implies that the price of EF1 is in fact $\Theta(\sqrt{n})$.
Their algorithm works by starting with an optimal allocation, arranging the goods on a line so that each bundle in this allocation is connected, giving each agent her favorite good from her bundle, and updating the allocation by carefully assigning additional items so as to maintain EF1 and connectivity on the line.
Moreover, their algorithm can be extended to the more general setting where agents have subadditive utilities.

\section{Preliminaries}
\label{sec:prelim}

Denote by $N=\{1,2,\dots,n\}$ the set of agents and $M=\{1,2,\dots,m\}$ the set of goods. Each agent $i$ has a nonnegative utility $u_i(j)$ for each good $j$. The agents' utilities are additive, meaning that $u_i(M')=\sum_{j\in M'}u_i(j)$ for every agent $i$ and subset of goods $M'\subseteq M$. Following \citet{CaragiannisKaKa12}, we normalize the utilities across agents by assuming that $u_i(M)=1$ for all $i$. We refer to a setting with agents, goods, and utility functions as an \emph{instance}.
An \emph{allocation} is a partition of $M$ into bundles $(M_1,\dots,M_n)$ such that agent $i$ receives bundle $M_i$. The \emph{(utilitarian) social welfare} of an allocation $\mathcal{M}$ is defined as $\SW(\mathcal{M}) := \sum_{i=1}^n u_i(M_i)$. The optimal social welfare for an instance $I$, denoted by $\OPT(I)$, is the maximum social welfare over all allocations for this instance.

A \emph{property} $P$ is a function that maps every instance $I$ to a (possibly empty) set of allocations $P(I)$. Every allocation in $P(I)$ is said to satisfy property $P$.

We are now ready to define the price of fairness concepts.

\begin{definition}
For any given property $P$ of allocations and any instance, we define the \emph{price of P} for that instance to be the ratio between the optimal social welfare and the maximum social welfare over allocations satisfying $P$:
\[\text{Price of }P\text{ for instance }I = \frac{\OPT(I)}{\max_{\mathcal{M}\in P(I)}\SW(\mathcal{M})}.\]
The overall \emph{price of $P$} is then defined as the supremum price of fairness across all instances.

Similarly, the \emph{strong price of $P$} for a given instance is the ratio between the optimal social welfare and the \emph{minimum} social welfare over allocations satisfying $P$:
\[\text{Strong price of }P\text{ for instance }I  = \frac{\OPT(I)}{\min_{\mathcal{M}\in P(I)}\SW(\mathcal{M})}.\]
The overall \emph{strong price of $P$} is then defined as the supremum price of fairness across all instances.
\end{definition}

We will only consider properties $P$ such that $P(I)$ is nonempty for every instance $I$, so the (strong) price of fairness is always well-defined. With the exception of Theorem~\ref{thm:round-robin-PoS-m}, we will be interested in the price of fairness as a function of $n$, and assume that $m$ can be arbitrary.

Next, we define the fairness properties that we consider. The first two properties are relaxations of the classical envy-freeness notion.

\begin{definition}[EF1]
An allocation is said to satisfy \emph{envy-freeness up to one good (EF1)} if for every pair of agents $i,i'$, there exists a set $A_{i'}\subseteq M_{i'}$ with $|A_{i'}|\leq 1$ such that $u_i(M_i)\geq u_i(M_{i'}\backslash A_{i'})$.
\end{definition}

\begin{definition}[EFX]
An allocation is said to satisfy \emph{envy-freeness up to any good (EFX)} if for every pair of agents $i,i'$ and every good $g\in M_{i'}$, we have $u_i(M_i)\geq u_i(M_{i'}\backslash\{g\})$.
\end{definition}

It is clear that EFX imposes a stronger requirement than EF1. An EF1 allocation always exists \citep{LiptonMaMo04}, while for EFX the existence question is still unresolved \citep{CaragiannisKuMo16}. As such, we will only consider EFX in the case of two agents, for which existence is guaranteed \citep{PlautRo20}.\footnote{Recently, \citet{ChaudhuryGaMe20} showed that the existence is also guaranteed for three agents.}

The round-robin algorithm, which we describe below, always computes an EF1 allocation (see, e.g., \citep{CaragiannisKuMo16}).

\begin{definition}[RR]
The \emph{round-robin algorithm} works by arranging the agents in some arbitrary order, and letting the next agent in the order choose her favorite good from the remaining goods.\footnote{In case there are ties between goods, we may assume worst-case tie breaking, since it is possible to obtain an instance with infinitesimal difference in welfare and any desired tie-breaking between goods by slightly perturbing the utilities.} An allocation is said to satisfy \emph{round-robin (RR)} if it is the result of applying the algorithm with some ordering of the agents.
\end{definition}

Our next property is balancedness, which means that the goods are spread out among the agents as much as possible. Balancedness and similar cardinality constraints have been considered in recent work \citep{BiswasBa18}. In addition to satisfying EF1, an allocation produced by the round-robin algorithm is also balanced.

\begin{definition}[BAL]
An allocation is said to be \emph{balanced (BAL)}  if $||M_i|-|M_j||\leq 1$ for any $i,j$.
\end{definition}

Next, we define a number of welfare maximizers.

\begin{definition}[MNW]
The \emph{Nash welfare} of an allocation is defined as $\prod_{i\in N}u_i(M_i)$. An allocation is said to be a \emph{maximum Nash welfare (MNW)} allocation if it has the maximum Nash welfare among all allocations.\footnote{In the case where the maximum Nash welfare is 0, an allocation is an MNW allocation if it gives positive utility to a set of agents of maximal size and moreover maximizes the product of utilities of the agents in that set.}
\end{definition}

\begin{definition}[MEW]
The \emph{egalitarian welfare} of an allocation is defined as $\min_{i\in N}u_i(M_i)$. An allocation is said to be a \emph{maximum egalitarian welfare (MEW)} allocation if it has the maximum egalitarian welfare among all allocations.
\end{definition}

\begin{definition}[LEX]
An allocation is said to be \emph{leximin (LEX)} if it maximizes the lowest utility (i.e., the egalitarian welfare), and, among all such allocations, maximizes the second lowest utility, and so on.
\end{definition}

Finally, we define Pareto optimality. While this is an efficiency notion rather than a fairness notion, we also consider it as it is a fundamental property in the context of resource allocation.

\begin{definition}[PO]
Given an allocation $(M_1,\dots,M_n)$, another allocation $(M_1',\dots,M_n')$ is said to be a \emph{Pareto improvement} if $u_i(M_i')\geq u_i(M_i)$ for all $i$ with at least one strict inequality. An allocation is \emph{Pareto optimal (PO)} if it does not admit a Pareto improvement.
\end{definition}

\citet{CaragiannisKuMo16} showed that an MNW allocation always satisfies EF1 and Pareto optimality. It is clear from the definition that any leximin allocation is Pareto optimal and maximizes egalitarian welfare. The problem of computing an MEW allocation has been studied by \citet{BezakovaDa05} and \citet{BansalSr06}. Leximin allocations were studied by \citet{BogomolnaiaMo04} and shown to be applicable in practice by \citet{KurokawaPrSh15}.

%All omitted proofs can be found in the online appendix.

\section{Envy-Freeness Relaxations and Round-Robin Algorithm}

In this section, we consider envy-freeness relaxations and the round-robin algorithm, which always produces an EF1 allocation. 

\subsection{Envy-Freeness Relaxations}

We begin with a general lower bound on the price of EF1.

\begin{theorem}
\label{thm:EF1-PoS-lower}
The price of EF1 is $\Omega(\sqrt{n})$.
\end{theorem}

\begin{proof}
Let $m=n, r=\lfloor\sqrt{n}\rfloor$, and assume that the utilities are as follows:
\begin{itemize}
    \item For $i=1,\dots,r-1$: $u_i((i-1)r+j)=\frac{1}{r}$ for $j=1,\dots,r$, and $u_i(j)=0$ otherwise.
    \item $u_r(j)=\frac{1}{n-r(r-1)}$ for $j=r(r-1)+1,\dots,n$, and $u_r(j)=0$ otherwise.
    \item For $i=r+1,\dots,n$: $u_i(j)=\frac{1}{n}$ for all $j$.
\end{itemize}

Consider the allocation that assigns goods $ir-r+1,\dots,ir$ to agent $i$ for $i=1,\dots,r-1$ and the remaining goods to agent $r$. The social welfare of this allocation is $r$. On the other hand, in any EF1 allocation, each of the agents $i=r+1,\dots,n$ must receive at least one good---otherwise some agent would receive at least two goods and agent $i$ would envy her. This means that the social welfare is at most $r\cdot\frac{1}{r}+(n-r)\cdot\frac{1}{n}<2$. Hence the price of EF1 is at least $\frac{r}{2}=\frac{\lfloor\sqrt{n}\rfloor}{2}$.
\end{proof}

For two agents, we establish an almost tight bound on the price of EF1 and a tight bound on the price of EFX.
We start with a lower bound for EF1.

\begin{theorem}
\label{thm:EF1-PoS-2-lower}
For $n=2$, the price of EF1 is at least $\frac87 \approx 1.143$.
\end{theorem}

\begin{proof}
Let $m=3$ and $0<\epsilon<1/6$, and assume that the utilities are as follows:
\begin{itemize}
    \item $u_1(1) = 1/3 - 2\epsilon, u_1(2) = 1/3 + \epsilon, u_1(3) = 1/3 + \epsilon$
    \item $u_2(1) = 0, u_2(2) = 1/2, u_2(3) = 1/2$
\end{itemize}
The optimal social welfare is $4/3 - 2\epsilon$, achieved by assigning the first good to agent 1 and the last two goods to agent 2. However, in any EF1 allocation the last two goods cannot both be given to agent 2. Hence the social welfare of an EF1 allocation is at most $(1/3 - 2\epsilon) + (1/3 + \epsilon) + 1/2 = 7/6 - \epsilon$. Taking $\epsilon \rightarrow 0$, we find that the price of EF1 is at least $\frac{4/3}{7/6} = 8/7$.
\end{proof}

We now turn to the upper bound.
In order to construct an EF1 allocation with high welfare, we proceed in a similar manner to the \emph{adjusted winner} procedure \citep{BramsTa96}, which is used to allocate divisible goods between two agents.
Specifically, we arrange the goods according to the ratios between the utilities that they yield for the two agents---the idea is that the agents will then prefer goods at different ends.
Roughly speaking, we then let the agent who obtains a lower utility in an optimal allocation choose a minimal set of goods for which she is EF1 starting from her end.

\begin{theorem}
\label{thm:EF1-PoS-2-upper}
For $n=2$, the price of EF1 is at most $\frac{2}{\sqrt{3}} \approx 1.155$.
\end{theorem}

\begin{proof}
Consider an arbitrary instance. Sort the goods so that $\frac{u_1(1)}{u_2(1)} \geq \frac{u_1(2)}{u_2(2)} \geq \cdots \geq \frac{u_1(m)}{u_2(m)}$; goods $x$ such that $u_2(x) = 0$ are put at the front and those with $u_1(x) = 0$ at the back, with arbitrary tie-breaking within each group of goods. (Goods that yield zero value to both agents can be safely ignored since they have no effect on the optimal welfare or the maximum welfare of an EF1 allocation.) For ease of notation, for any $1 \leq k \leq m$ we write $L(k) := \{1, \ldots, k\}$ and $R(k) := \{k, \ldots, m\}$. We also define $L(0) = R(m+1) = \emptyset$.

%Let $S_1 = \{i \mid \frac{u_1(i)}{u_2(i)} \geq 1\} = L(s)$ for some $1 \leq s \leq m$ and $S_2 = M \backslash S_1 = R(s+1)$. If $s=m$, both agents have identical valuations and the price of EF1 is $1$, so we may assume that $s<m$.
Let $S_1 := \{i \mid \frac{u_1(i)}{u_2(i)} > 1\} = L(s)$ for some $0 \leq s \leq m$ and $S_2 := M\setminus S_1 = R(s+1)$. It is easy to see that $s<m$. If $s=0$, both agents have identical utilities and the price of EF1 is $1$, so we may assume that $s>0$. The allocation $\mathcal{S} = (S_1, S_2)$ is an optimal allocation, and the optimal social welfare is $u_1(S_1) + u_2(S_2)$.
Without loss of generality, assume that $u_1(S_1) \leq u_2(S_2)$. Note that we must have $u_2(S_2) \geq \frac12$, since otherwise both $u_1(S_1)$ and $u_2(S_2)$ are smaller than $\frac12$ and switching $S_1$ and $S_2$ would yield a higher social welfare. We can further assume that $u_1(S_1) < \frac{1}{2}$, because otherwise $\mathcal{S}$ is also an EF1 allocation and the price of EF1 is $1$.

Next, we describe how to obtain a particular EF1 allocation $\mathcal{F}$.
Let $f$ be the smallest index such that $f \geq s$ and $u_1(L(f)) \geq u_1(R(f+2))$. Clearly, $f<m$. In the allocation $\mathcal{F} = (F_1, F_2)$, we assign the goods $F_1 := L(f)$ to agent 1, and $F_2 := R(f+1)$ to agent 2.
The pseudocode for computing $\mathcal{F}$ is presented as Algorithm~\ref{alg:EF1}.
See also Figure~\ref{fig:EF1}.

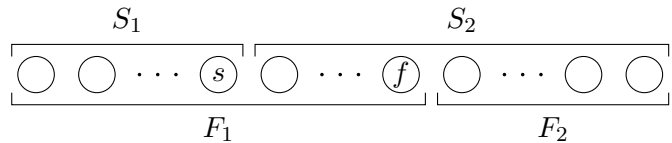
\begin{figure}[!ht]
\centering
\begin{tikzpicture}[scale=0.8]
\draw (3,3) circle [radius = 0.3];
\draw (4,3) circle [radius = 0.3];
\draw[fill] (4.7,3) circle [radius = 0.02];
\draw[fill] (5,3) circle [radius = 0.02];
\draw[fill] (5.3,3) circle [radius = 0.02];
\draw (6,3) circle [radius = 0.3];
\draw (7,3) circle [radius = 0.3];
\draw[fill] (7.7,3) circle [radius = 0.02];
\draw[fill] (8,3) circle [radius = 0.02];
\draw[fill] (8.3,3) circle [radius = 0.02];
\draw (9,3) circle [radius = 0.3];
\draw (10,3) circle [radius = 0.3];
\draw[fill] (10.7,3) circle [radius = 0.02];
\draw[fill] (11,3) circle [radius = 0.02];
\draw[fill] (11.3,3) circle [radius = 0.02];
\draw (12,3) circle [radius = 0.3];
\draw (13,3) circle [radius = 0.3];
\node at (6,3) {$s$};
\node at (9,3) {$f$};
\draw (2.6,3.3) -- (2.6,3.5) -- (6.4,3.5) -- (6.4,3.3);
\draw (6.6,3.3) -- (6.6,3.5) -- (13.4,3.5) -- (13.4,3.3);
\draw (2.6,2.7) -- (2.6,2.5) -- (9.4,2.5) -- (9.4,2.7);
\draw (9.6,2.7) -- (9.6,2.5) -- (13.4,2.5) -- (13.4,2.7);
\node at (4.5,3.9) {$S_1$};
\node at (10,3.9) {$S_2$};
\node at (6,2.1) {$F_1$};
\node at (11.5,2.1) {$F_2$};
\end{tikzpicture}
\caption{Illustration for the proof of Theorem~\ref{thm:EF1-PoS-2-upper}.}
\label{fig:EF1}
\end{figure}

\begin{algorithm}
\caption{For computing an EF1 allocation with high social welfare between two agents}\label{alg:EF1}
\begin{algorithmic}[1]
\Procedure{EF1-Two-Agents}{$N,M,u_1,u_2$}
\State Assume that in an optimal allocation, agent $1$ obtains no higher utility than agent $2$. (Otherwise, reverse the roles of the two agents.)
\State Sort the goods so that $\frac{u_1(1)}{u_2(1)} \geq \frac{u_1(2)}{u_2(2)} \geq \cdots \geq \frac{u_1(m)}{u_2(m)}$.
\For{$k=1,2,\dots,m$}
\State $L(k)\leftarrow \{1,\dots,k\}$
\State $R(k)\leftarrow \{k,\dots,m\}$
\EndFor
\State $s\leftarrow 0$
\While{$s < m$ and $\frac{u_1(s+1)}{u_2(s+1)} > 1$}
\State $s\leftarrow s+1$
\EndWhile
\State $f\leftarrow s$
\While{$u_1(L(f)) < u_1(R(f+2))$}
\State $f\leftarrow f+1$
\EndWhile
\State \Return $(L(f),R(f+1))$
\EndProcedure
\end{algorithmic}
\end{algorithm}

\paragraph{Allocation $\mathcal{F}$ satisfies EF1.}

The EF1 condition is satisfied for agent 1, because $u_1(F_1) \geq u_1(F_2 \backslash \{f+1\})$ by definition.

    For agent 2, since $f$ is the smallest index such that $f \geq s$ and $u_1(L(f)) \geq u_1(R(f+2))$, we have either $f = s$ or $u_1(L(f-1)) < u_1(R(f+1))$.

    If $f = s$, then $\mathcal{F}$ coincides with the optimal allocation $\mathcal{S}$, and $u_2(F_2) = u_2(S_2) \geq \frac12$. Clearly EF1 is satisfied.

    Else, $f > s$, and we have $0<u_1(L(f-1)) < u_1(R(f+1))$. Note also that $u_2(R(f+1))>0$. Therefore, 
    \begin{align*}
    \frac{u_1(L(f-1))}{u_2(L(f-1))} &\geq \frac{u_1(f-1)}{u_2(f-1)} 
    \geq \frac{u_1(f+1)}{u_2(f+1)} \geq \frac{u_1(R(f+1))}{u_2(R(f+1))},    
    \end{align*}    
    where we take a fraction to be infinite if it has denominator $0$.\footnote{To see the first and third inequalities, one may prove by induction that in general, if we have $\frac{a_1}{b_1}\geq\dots\geq\frac{a_k}{b_k}$, then $\frac{a_1}{b_1}\geq\frac{a_1+\dots+a_k}{b_1+\dots+b_k}\geq \frac{a_k}{b_k}$. The case $k=2$ holds because $\frac{a_1}{b_1}\geq\frac{a_1+a_2}{b_1+b_2}$ is equivalent to $\frac{a_1}{b_1}\geq\frac{a_2}{b_2}$, and similarly for $\frac{a_1+a_2}{b_1+b_2}\geq\frac{a_2}{b_2}$.} (None of the fractions can have both numerator and denominator $0$.)
    
    Since $u_1(L(f-1)) < u_1(R(f+1))$, this implies that
    $$\frac{u_2(L(f-1))}{u_2(R(f+1))} \leq \frac{u_1(L(f-1))}{u_1(R(f+1))} < 1.$$
    Thus, $$u_2(F_2) = u_2(R(f+1)) > u_2(L(f-1)) = u_2(F_1 \backslash \{f\}),$$
    implying that EF1 is again satisfied.

\paragraph{The price of EF1 for this instance is at most $\frac{2}{\sqrt{3}}$.}
Now we analyze the social welfare of the allocation $\mathcal{F}$ and compare it to the optimal social welfare.

If $f=s$, the price of EF1 is $1$. Assume from now on that $f>s$. We have $u_1(F_2) > u_1(L(f-1)) \geq u_1(L(s)) = u_1(S_1)$
and $\frac{u_1(S_2)}{u_2(S_2)} \geq \frac{u_1(F_2)}{u_2(F_2)}$. Since $u_1(F_2)>0$, we also have $u_1(S_2)>0$. Moreover, $u_2(F_2),u_2(S_2)>0$. 
Thus,
\begin{align*}
u_1(F_1) + u_2(F_2) & \geq  (1 - u_1(F_2)) + \frac{u_1(F_2)u_2(S_2)}{u_1(S_2)} \\
& =  1 + \left(\frac{u_2(S_2)}{u_1(S_2)}-1\right)u_1(F_2) \\
& > 1 + \left(\frac{u_2(S_2)}{u_1(S_2)}-1\right)u_1(S_1) \\
& = 1 - u_1(S_1) + \frac{u_2(S_2)}{u_1(S_2)}\cdot u_1(S_1) \\
& = 1 - u_1(S_1) + \frac{u_2(S_2)}{1-u_1(S_1)}\cdot(1+(u_1(S_1)-1)) \\
& = 1 - u_1(S_1) + \frac{u_2(S_2)}{1-u_1(S_1)} - u_2(S_2).
\end{align*}
Therefore the ratio between the optimal social welfare and the social welfare of $\mathcal{F}$ is
\begin{align*}
\alpha &:= \frac{u_1(S_1)+u_2(S_2)}{u_1(F_1)+u_2(F_2)} < \frac{u_1(S_1)+u_2(S_2)}{\frac{u_2(S_2)}{1 - u_1(S_1)} + 1 - u_2(S_2) - u_1(S_1)}.
\end{align*}

We further analyze the last expression. First, taking its partial derivative with respect to $u_2(S_2)$ gives
\[
\frac{(1-u_1(S_1))(1-2u_1(S_1))}{(u_1(S_1)^2+u_1(S_1)(u_2(S_2)-2)+1)^2},
\]
which is always positive when $u_1(S_1) < \frac12$.
This shows that the last expression is monotone increasing in $u_2(S_2)$. Thus
$$\alpha < \frac{u_1(S_1) + 1}{\frac{1}{1 - u_1(S_1)} - u_1(S_1)}.$$
Finally, this expression is maximized when $u_1(S_1) = 2-\sqrt{3}$ and yields a value of $\frac{2}{\sqrt{3}}$, completing the proof.
\end{proof}

The gap on the price of EF1 between Theorems~\ref{thm:EF1-PoS-2-lower} and \ref{thm:EF1-PoS-2-upper} is only approximately $0.01$.
For EFX, we establish a tight bound in the case of two agents.

\begin{theorem}
\label{thm:EFX-PoS-2}
For $n=2$, the price of EFX is $3/2$.
\end{theorem}

\begin{proof}
\emph{Lower bound}: Let $m=3$ and $0<\epsilon<1/2$, and assume that the utilities are as follows:
\begin{itemize}
    \item $u_1(1)=1/2+\epsilon, u_1(2)=1/2-\epsilon, u_1(3)=0$.
    \item $u_2(1)=1/2+\epsilon, u_2(2)=0, u_2(3)=1/2-\epsilon$.
\end{itemize}

The optimal social welfare is $3/2-\epsilon$, achieved by assigning the first two goods to agent 1 and the last good to agent 2.  On the other hand, in any EFX allocation, no agent can get both of the goods that they positively value. Hence, the social welfare of an EFX allocation is at most $1$. Taking $\epsilon\rightarrow 0$, we find that the price of EFX is at least $3/2$.

\medskip

\noindent\emph{Upper bound}: Consider an arbitrary instance. If in an optimal allocation both agents get utility at least $1/2$, this allocation is also envy-free and hence EFX, so the price of EFX is $1$. Otherwise, the maximum welfare is at most $1+1/2=3/2$. Now we show that there always exists an EFX allocation with social welfare at least $1$; this immediately yields the desired bound.

Let the first agent partition the goods into two bundles such that her values for the bundles are as equal as possible. Denote by $x$ and $1-x$ the values of the two bundles, where $x\geq 1-x$. Suppose that all goods of zero value, if any, are in the second bundle. Let $y\geq 1-y$ be the corresponding values for the second agent, and assume without loss of generality that $y\geq x$. Consider the partition of the first agent, and assume that the two bundles yield value $z$ and $1-z$ to the second agent, respectively. If $z\leq 1-z$, by assigning the first bundle to the first agent and the second bundle to the second agent, we have an envy-free allocation with welfare at least $1$. Else, $z\geq 1-z$. By definition of $y$, we also have $z\geq y\geq x$. We assign the first bundle to the second agent and the second bundle to the first agent. The second agent is clearly envy-free. If the first agent still has envy after removing some good $i$ from the first bundle, then by moving good $i$ to the second bundle, we create a more equal partition, a contradiction. Hence the allocation is EFX to the first agent. The social welfare of this allocation is $z+(1-x)\geq 1$.
\end{proof}

Next, we give a simple instance showing that EF1 and EFX allocations can have arbitrarily bad welfare. 

\begin{theorem}
\label{thm:EF1-EFX-PoA}
The strong price of EF1 is $\infty$. For $n=2$, the strong price of EFX is $\infty$.
\end{theorem}

\begin{proof}
Let $m=n$, and assume that $u_i(i)=1$ for all $i$ and $u_i(j)=0$ otherwise. The allocation that assigns good $i$ to agent $i$ for every $i$ has social welfare $n$. On the other hand, the allocation that assigns good $i-1$ to agent $i$ for $i=2,\dots,n$ and good $n$ to agent $1$ is EF1 and EFX, but has social welfare $0$. The conclusion follows.
\end{proof}

\subsection{Round-Robin Algorithm}

We now turn our attention to the round-robin algorithm. We show that it is always possible to order the agents to obtain a welfare of $1$.

\begin{lemma}
\label{lem:roundrobin-welfare}
For any instance, there exists an ordering of the agents such that the round-robin algorithm implemented with this ordering produces an allocation with social welfare at least $1$, and this bound is tight.
\end{lemma}

\begin{proof}
We claim that if we choose the ordering of the agents uniformly at random, the expected social welfare is at least $1$. The desired bound immediately follows from this claim.

To prove the claim, consider an arbitrary agent $i$, and assume without loss of generality that $u_i(1)\geq u_i(2)\geq\ldots\geq u_i(m)$. Note that if the agent is ranked $j$th in the ordering, her utility  is at least $u_i(j)+u_i(n+j)+u_i(2n+j)+\dots+u_i(kn+j)$, where $k=\lfloor(m-j)/n\rfloor$. Hence, the agent's expected utility is at least $$\frac{1}{n}\cdot\sum_{j=1}^n\sum_{r=0}^{\left\lfloor (m-j)/n\right\rfloor}u_i(rn+j) = \frac{1}{n}\cdot\sum_{j=1}^m u_i(j) = \frac{1}{n}.$$ 
It follows from linearity of expectation that the expected social welfare is at least $n\cdot\frac{1}{n}=1$, as claimed.

The tightness of the bound follows from the instance where every agent has utility $1$ for the same good.
\end{proof}

Lemma~\ref{lem:roundrobin-welfare} yields a linear price of fairness for round-robin.

\begin{theorem}
\label{thm:round-robin-PoS}
The price of round-robin is $n$. Consequently, the price of EF1 is at most $n$.
\end{theorem}

\begin{proof}
\emph{Upper bound}: Consider an arbitrary instance. Since every agent receives utility at most $1$, the optimal social welfare is at most $n$. On the other hand, by Lemma~\ref{lem:roundrobin-welfare}, there exists an ordering of the agents such that the round-robin algorithm yields welfare at least $1$. Hence the price of round-robin is at most $n$.

\medskip

\noindent\emph{Lower bound}: Let $m = x^n$ for some large $x$ that is divisible by $n$, and assume that the utilities are such that for each agent $i$, $u_i(j)=1/x^i$ for $j=1,\dots,x^i$ and $u_i(j)=0$ otherwise.

Consider the allocation that assigns goods $1,\dots,x$ to agent 1, and $x^{i-1}+1,\dots,x^i$ to agent $i$ for every $i\geq 2$. In this allocation, agent 1 gets utility 1, while each remaining agent gets utility $(x^i - x^{i-1})/x^i = 1 - 1/x$. The social welfare is therefore $n - (n-1)/x$. This converges to $n$ for large $x$.

On the other hand, consider the round-robin algorithm with an arbitrary ordering of the agents, and assume without loss of generality that agents always break ties in favor of goods with lower numbers. Hence, regardless of the ordering, the goods get chosen in the order $1,2,\dots,m$. As a result, every agent gets exactly $1/n$ of their valued goods, so her utility is $1/n$, and the social welfare is $1$. Hence the price of round-robin is $n$.
\end{proof}

The argument for the lower bound in Theorem~\ref{thm:round-robin-PoS} works even if we can choose a new ordering of the agents in every round. This means that the fixed order is not a barrier to obtaining a better price of fairness, but rather the ``each agent picks exactly once in every round'' aspect of the algorithm. 

One may notice that the lower bound construction uses an exponential number of goods. This is in fact necessary to obtain an instance with a high price of round-robin. As we show next, the $\Omega(\sqrt{n})$ lower bound on the price of EF1 is almost tight as long as $m$ is not too large compared to $n$.
At a high level, our proof proceeds by considering an optimal allocation and choosing a range $[2^{-\ell-1},2^{-\ell}]$ that the largest number of agents' utilities for goods in this allocation fall into.
In the case where a sufficiently large number of goods correspond to this range, we may choose an arbitrary round-robin ordering---we can lower bound the welfare resulting from the round-robin algorithm by observing that as long as we have not run out of goods from this range with respect to an agent, every pick must give the agent a utility at least the minimum utility that the agent obtains from this range.
On the other hand, if only a small number of goods belong to this range, we need to be more careful in choosing the ordering.

\begin{theorem}
\label{thm:round-robin-PoS-m}
The price of round-robin is $O(\sqrt{n}\log(mn))$. Consequently, the price of EF1 is $O(\sqrt{n}\log(mn))$.
\end{theorem}

\begin{proof}
%Notice that when $m < 2$, the price of round-robin is obviously $1$. Thus, we may assume that $m \geq 2$.

Consider any instance $I$. 
First, observe that if $\OPT(I) \leq 65\sqrt{n}\log_2 (mn)$, then Lemma~\ref{lem:roundrobin-welfare} immediately yields the desired result.
We therefore focus on the case where $\OPT(I) > 65\sqrt{n}\log_2 (mn)$.
We claim that there exists an ordering for which the round-robin algorithm produces an allocation with social welfare at least $\frac{\OPT(I)}{65\sqrt{n} \log_2 (mn)}$. 

Fix an optimal allocation $\mathcal{M} = (M_1, \dots, M_n)$, and let $r:=\lceil \log_2(m\sqrt{n}) \rceil$. For each $i \in N$, let us partition $M_i$ into $M_i^0 \cup M_i^1 \cup \dots\cup M_i^r$, where $M_i^\ell$ is defined by
\begin{align*}
M_i^\ell =
\begin{cases}
\{j \in M_i \mid u_i(j) \in (2^{-\ell - 1}, 2^{-\ell}]\} & \text{if } \ell \ne r; \\
\{j \in M_i \mid u_i(j) \in [0, 2^{-\ell}]\} & \text{if } \ell = r.
\end{cases}
\end{align*}
Furthermore, define $M^\ell := \cup_{i=1}^n M_i^\ell$ and $\SW_\ell(\mathcal{M}) := \sum_{i=1}^n u_i(M^\ell_i)$.  

Let $\ell^*:=\argmax_{\ell \in \{0, \dots, r - 1\}} \SW_\ell(\cM)$. We have
\begin{align*}
\SW_{\ell^*}(\cM) \geq \frac{1}{r} \left(\sum_{\ell=0}^{r-1}\SW_\ell(\cM)\right)
= \frac{\OPT(I) - \SW_r(\cM)}{r}.
\end{align*}
However, since agent $i$ values each item in $M_i^r$ at most $2^{-r} \leq \frac{1}{m\sqrt{n}}$, we have $u_i(M_i^r) \leq 1/\sqrt{n}$. This implies that $\SW_r(\cM) \leq \sqrt{n}$, which is no more than $\OPT(I)/65$. Hence,
\begin{align}
\label{eq:highwelfare}
\SW_{\ell^*}(\cM) \geq \frac{64}{65r} \cdot \OPT(I) \geq \frac{32\cdot\OPT(I)}{65\log_2(mn)}.
\end{align}
Thus, it suffices to show the existence of an ordering such that round-robin produces an allocation with social welfare at least $\SW_{\ell^*}(\cM)/\sqrt{n}$.

Observe that (\ref{eq:highwelfare}) implies that $\SW_{\ell^*}(\cM) > 32\sqrt{n}$.
We now consider two cases, based on $T := |M^{\ell^*}|$. Since $u_i(M_i^{\ell^*}) \leq 2^{-\ell^*}|M_i^{\ell^*}|$ for each $i$, we have $\SW_{\ell^*}(\cM) \leq 2^{-\ell^*}T$.

\underline{Case 1}: $T > 2n$. In this case, we will show that the round-robin algorithm with arbitrary ordering yields an allocation with social welfare at least $\SW_{\ell^*}(\cM)/\sqrt{n}$. 

To see this, let us consider the round-robin procedure with arbitrary ordering, and consider the set of goods that are picked in the first $t := \lfloor T/(2n) \rfloor$ rounds; let $S_t \subseteq M$ denote this set. Now, observe that 
\begin{align*}
\sum_{i=1}^n |M_i^{\ell^*} \setminus S_t| \geq T - |S_t| = T - n \cdot t \geq \frac{T}{2}.
\end{align*}
This implies that
\begin{align*}
\sum_{i=1}^n u_i(M_i^{\ell^*} \setminus S_t) \geq \frac{T}{2} \cdot 2^{-\ell^* - 1} \geq \frac{\SW_{\ell^*}(\cM)}{4} > 8\sqrt{n}.
\end{align*}
Since $u_i(M_i^{\ell^*} \setminus S_t) \leq 1$, there must be more than $8\sqrt{n}$ agents such that $M_i^{\ell^*} \nsubseteq S_t$. Let $N^*$ denote the set of such agents.

We claim that, in each of the first $t$ rounds, every agent $i \in N^*$ must receive an item she values at least $2^{-\ell^* - 1}$. The reason is that agent $i$ picks her favorite good, which she must value at least as much as the good(s) left unpicked in $M_i^{\ell^*} \setminus S_t$. Moreover, she values the latter at least $2^{-\ell^* - 1}$, so this must also be a lower bound of her utility for the former.

From the claim in the previous paragraph, we can conclude that the social welfare of the allocation produced is at least
\begin{align*}
|N^*| \cdot t \cdot 2^{-\ell^* - 1} > 8\sqrt{n} \cdot \frac{T}{4n} \cdot 2^{-\ell^* - 1} \geq \frac{\SW_{\ell^*}(\cM)}{\sqrt{n}}
\end{align*}
as desired. Note that we use the assumption $T > 2n$ to conclude that $t \geq T/(4n)$ in the first inequality above.

\underline{Case 2}: $T \leq 2n$. In this case, we will show that if we choose the ordering $\pi$ in a careful manner, then the social welfare obtained in the first round alone already suffices.

Similarly to Case 1, observe that since $\sum_{i=1}^n u_i(M_i^{\ell^*}) = \SW_{\ell^*}(\cM) > 8\sqrt{n}$, there are more than $8\sqrt{n}$ agents $i$ whose $M_i^{\ell^*}$ is non-empty. Let $N^*$ denote the set of such agents. 

We will construct the ordering $\pi$ step-by-step as follows. For $k = 1, \dots, \lceil 4\sqrt{n} \rceil$, we let $\pi(k)$ be any agent $i$ such that (1) $i$ is not yet in the ordering and (2) not all goods in $M_i^{\ell^*}$ are already picked by $\pi(1), \dots, \pi(k - 1)$. Note that such an agent exists because, at each step $k$, at most two candidate agents become invalid: the agent $i = \pi(k)$, and the agent $i'$ whose good in $M^{\ell^*}_{i'}$ is picked by $\pi(k)$. Since we start with $8\sqrt{n}$ valid candidates, even after $\lceil 4\sqrt{n} \rceil - 1$ steps, there are still valid candidate agents to be chosen from.

The remainder of the ordering can be chosen arbitrarily. We now argue that the resulting round-robin allocation has the desired social welfare. To see this, for $k = 1, \dots, \lceil 4\sqrt{n} \rceil$, observe that agent $\pi(k)$ must pick a good that is worth at least $2^{-\ell^* - 1}$ to her in the first round, since not all goods in $M_{\pi(k)}^{\ell^*}$ have been picked. As a result, the social welfare is at least
\begin{align*}
(4\sqrt{n}) \cdot 2^{-\ell^* - 1} \geq (2T/\sqrt{n}) \cdot 2^{-\ell^* - 1}  \geq \frac{\SW_{\ell^*}(\cM)}{\sqrt{n}},
\end{align*}
where the first inequality follows from $T \leq 2n$.
\end{proof}

We end this section by establishing an exact bound on the strong price of round-robin.

\begin{theorem}
\label{thm:round-robin-PoA}
The strong price of round-robin is $n^2$.
\end{theorem}

\begin{proof}
\emph{Upper bound}: Consider an arbitrary instance. Since every agent receives utility at most $1$, the optimal social welfare is at most $n$. On the other hand, in the round-robin algorithm, the first agent gets to choose an item ahead of all other agents in every round and therefore does not envy any other agent in the resulting allocation. This implies that her utility, and hence the social welfare, is at least $1/n$. It follows that the strong price of round-robin is at most $n^2$.

\medskip

\noindent\emph{Lower bound}: Let $m$ be a large number divisible by $n$, and assume that the utilities are as follows:
\begin{itemize}
    \item $u_1(i)=\frac{1}{m}$ for all $i$.
    \item For $i=2,\dots,n$: $u_i(i-1)=1$, and $u_i(j)=0$ otherwise.
\end{itemize}

Consider the allocation that assigns good $i-1$ to agent $i$ for every $i=2,\dots,n$, and the remaining goods to agent 1. In this allocation, every agent $i\geq 2$ receives utility $1$. Agent 1 receives utility $\frac{m-n+1}{m}$, which converges to $1$ for large $m$. Therefore the social welfare converges to $n$.

On the other hand, consider the round-robin algorithm with the ordering of the agents $1,\dots,n$, and assume without loss of generality that agents always break ties in favor of goods with lower numbers. The first agent gets utility exactly $1/n$, while the remaining agents get zero utility since their only valuable good is ``stolen'' by the agent before them in the first round. Hence the social welfare is $1/n$. This means that the strong price of round-robin is $n^2$, as desired.
\end{proof}

\section{Balancedness}

In this section, we consider balancedness. We begin by establishing an asymptotically tight bound on the price of balancedness.

\begin{theorem}
\label{thm:balanced-PoS}
The price of balancedness is $\Theta(\sqrt{n})$.
\end{theorem}

\begin{proof}
\emph{Lower bound}: Consider the instance in Theorem~\ref{thm:EF1-PoS-lower}. The social welfare can be as high as $r=\lfloor\sqrt{n}\rfloor$, while a similar argument shows that the social welfare of any balanced allocation is at most $2$. The conclusion follows.

\medskip

Intuitively, for the upper bound, we divide the agents into two groups according to whether they receive a sufficiently large number of goods in an optimal allocation or not.
There cannot be too many agents in the first group, and therefore they cannot make a significant contribution to the optimal welfare, so we may ignore them.
For agents in the second group, we let each of them keep some number of goods that they like most; we choose this number so that it is possible to redistribute the remaining goods and obtain a balanced allocation.

\medskip

\noindent\emph{Upper bound}: 
If $\OPT(I)\leq 4\sqrt{n}$, the result follows immediately from Lemma~\ref{lem:roundrobin-welfare}.
We therefore assume that $\OPT(I)> 4\sqrt{n}$.
We claim that for any instance $I$, the maximum social welfare of a balanced allocation is always within a factor $4\sqrt{n}$ of the optimal social welfare; this claim implies the desired upper bound.    In fact, we will show that there is a balanced allocation $\mathcal{M}$ such that $\SW(\mathcal{M})\geq\frac{\OPT(I)-\sqrt{n}}{2\sqrt{n}}$; this suffices for our claim because $\frac{\OPT(I)-\sqrt{n}}{2\sqrt{n}}\geq\frac{\OPT(I)}{4\sqrt{n}}$. We consider two cases.

\underline{Case 1}: $m\geq n$. Fix an optimal allocation, and let $A$ be the set of agents who receive at least $\frac{m}{\sqrt{n}}$ goods in the optimal allocation, and $B$ the complement set of agents. Since there are at most $\sqrt{n}$ agents in $A$, they contribute at most $\sqrt{n}$ to $\OPT(I)$, so the agents in $B$ contribute at least $\OPT(I)-\sqrt{n}$. We let each agent in $B$ keep her $\left\lceil\frac{m}{2n}\right\rceil$ most valuable goods (or all of her goods, if she has fewer than this number of goods).
Note that each such agent keeps at least a  $\left\lceil\frac{m}{2n}\right\rceil/\frac{m}{\sqrt{n}} \geq \frac{1}{2\sqrt{n}}$ fraction of her goods.
Since the agents in $B$ originally have a total utility of at least $\OPT(I)-\sqrt{n}$, the utility of the kept goods is at least $\frac{\OPT(I)-\sqrt{n}}{2\sqrt{n}}$. Moreover, since $\left\lceil\frac{m}{2n}\right\rceil \leq \left\lfloor\frac{m}{n}\right\rfloor$ due to the assumption $m\geq n$, the remaining goods can be reallocated to obtain a balanced allocation, which has social welfare at least $\frac{\OPT(I)-\sqrt{n}}{2\sqrt{n}}$.

\underline{Case 2}: $m<n$. Fix an optimal allocation, and let $A$ be the set of agents who receive at least $\sqrt{n}$ goods in the optimal allocation, and $B$ the complement set of agents. Since there are at most $\sqrt{n}$ agents in $A$, they contribute at most $\sqrt{n}$ to $\OPT(I)$, so the agents in $B$ contribute at least $\OPT(I)-\sqrt{n}$.
We let each agent in $B$ keep her most valuable good (if she receives at least one good). By a similar reasoning as in Case~1, this yields a total utility of at least $\frac{\OPT(I)-\sqrt{n}}{\sqrt{n}}$. The remaining goods can be reallocated to obtain a balanced allocation, which has social welfare at least $\frac{\OPT(I)-\sqrt{n}}{\sqrt{n}}\geq \frac{\OPT(I)-\sqrt{n}}{2\sqrt{n}}$.
\end{proof}

For two agents, we give an exact bound on the welfare that can be lost due to imposing balancedness.

\begin{theorem}
\label{thm:balanced-PoS-2}
For $n=2$, the price of balancedness is $4/3$.
\end{theorem}

\begin{proof}
\emph{Lower bound}: Let $m$ be a large even number, and assume that the utilities are as follows:
\begin{itemize}
    \item $u_1(1)=1$ and $u_1(i)=0$ otherwise.
    \item $u_2(i)=\frac{1}{m}$ for all $i$.
\end{itemize}
Consider the allocation that assigns the first good to the first agent and the remaining goods to the second agent. The social welfare is $1+(1-1/m)$, which converges to $2$ for large $m$. On the other hand, in any balanced allocation, the first agent gets utility at most $1$ while the second agent gets utility $\frac{m}{2}\cdot\frac{1}{m}=\frac{1}{2}$, so the social welfare is at most $3/2$. Hence the price of balancedness is at least $4/3$.

\medskip

\noindent\emph{Upper bound}: Consider an arbitrary instance. If $m$ is odd, we may add a dummy good that yields zero utility to both agents---this does not change the optimal social welfare or the maximum social welfare of a balanced allocation. We may therefore assume that $m$ is even.

Sort the goods so that $u_1(1)-u_2(1)\geq u_1(2)-u_2(2)\geq\dots\geq u_1(m)-u_2(m)$. Let $s$ be the last good such that $u_1(s)-u_2(s)\geq 0$, and assume without loss of generality that $s\geq m/2$. An optimal allocation assigns the set of goods $S_1=\{1,\dots,s\}$ to the first agent and the complement set $S_2$ to the second agent, yielding social welfare $u_1(S_1)+u_2(S_2)=u_1(S_1)+(1-u_2(S_1))=1+\Delta$, where $\Delta:=u_1(S_1)-u_2(S_1)\geq 0$. On the other hand, consider the balanced allocation that assigns goods $1,\dots,m/2$ to the first agent and the remaining goods to the second agent. Note that at most half of the goods in $S_1$ are reallocated to the second agent, and these are the goods with the lowest difference in utility between the two agents. Hence, the utility loss going from the first to the second allocation is at most $\Delta/2$, implying that the social welfare of the second allocation is at least $1+\frac{\Delta}{2}$. The price of balancedness is therefore at most
$$\sup_{0\leq\Delta\leq 1}\frac{1+\Delta}{1+\frac{\Delta}{2}}.$$
This ratio is increasing in $\Delta$ and reaches the maximum at $\Delta=1$, where its value is $4/3$, completing the proof.
\end{proof}

Finally, the same construction as in Theorem~\ref{thm:EF1-EFX-PoA} shows that balanced allocations can have arbitrarily bad welfare.

\begin{theorem}
The strong price of balancedness is $\infty$.
\end{theorem}

\section{Welfare Maximizers}

In this section, we consider allocations that maximize different measures of welfare. To start with, we show that every MNW and leximin allocation yields a decent welfare.

\begin{lemma}
\label{lem:MNW-leximin-welfare}
For any instance, every MNW allocation and every leximin allocation has social welfare at least $1$, and both bounds are tight.
\end{lemma}

\begin{proof}
We first establish the bound for MNW. Consider any MNW allocation where agent $i$ receives bundle $M_i$, and assume for contradiction that $\sum_{k=1}^n u_k(M_k)<1$. Fix any agent $i$.
Since the agent has utility $1$ for the entire set of items, we have $\sum_{k=1}^n u_i(M_k)=1$.
If $u_i(M_k)\leq u_k(M_k)$ for all $k=1,\dots,n$, we would have $$1 = \sum_{k=1}^n u_i(M_k)\leq \sum_{k=1}^n u_k(M_k)<1,$$
a contradiction, so there exists $j\neq i$ such that $u_i(M_j)>u_j(M_j)$. 
Construct a directed graph with vertices $1,2,\dots,n$, and add an edge from $i$ to $j$ if $u_i(M_j)>u_j(M_j)$. 
From the above observation, every vertex has at least one outgoing edge, implying that the graph consists of a directed cycle. For every edge $i\rightarrow j$ in the cycle, we give $M_j$ to agent $i$ instead of agent $j$. If we consider the change in the multiset of the $n$ utilities between the old and new allocations, at least one number increases while others remain the same. This means that either we have decreased the number of agents who get zero utility, or keep this number fixed and increase the product of utilities of the agents who get nonzero utility. Either case contradicts the definition of an MNW allocation.

To show the bound for leximin, we apply the same argument. An improvement in the multiset of utilities as described in the last step contradicts the definition of leximin.

Finally, the tightness of the bounds follows from the instance where every agent has utility $1$ for the same good.
\end{proof}

Lemma~\ref{lem:MNW-leximin-welfare} allows us to show that the price of MNW and the strong price of MNW, the price of MEW, and both prices of leximin are of linear order.

\begin{theorem}
\label{thm:MNW-PoS-PoA}
The price of MNW, the strong price of MNW, the price of MEW, the price of leximin, and the strong price of leximin are $\Theta(n)$.
\end{theorem}

\begin{proof}
We start with MNW. It suffices to show that the price of MNW is $\Omega(n)$ and the strong price of MNW is $O(n)$.

\medskip

\noindent\emph{Lower bound}: Let $m=n$ and $0<\epsilon<1$, and assume that the utilities are as follows:
\begin{itemize}
    \item $u_1(1)=1$ and $u_1(j)=0$ otherwise.
    \item For $i=2,\dots,n$: $u_i(i-1)=1-\epsilon$, $u_i(i)=\epsilon$, and $u_i(j)=0$ otherwise.
\end{itemize}

Consider the allocation that assigns good $i-1$ to agent $i$ for $i=2,\dots,n$, and good $n$ to agent 1. The social welfare of this allocation is $(n-1)(1-\epsilon)$. On the other hand, the unique MNW allocation assigns good $i$ to agent $i$ for every $i$. The social welfare of this allocation is $1+(n-1)\epsilon$. Taking $\epsilon\rightarrow 0$, we find that the price of MNW is $\Omega(n)$.

\medskip

\noindent\emph{Upper bound}: Consider an arbitrary instance. Since every agent receives utility at most $1$, the optimal social welfare is at most $n$. On the other hand, by Lemma~\ref{lem:MNW-leximin-welfare}, the social welfare of any MNW allocation is at least $1$. The conclusion follows.

\medskip

Notice that in the lower bound instance above, the unique MNW allocation is also the unique MEW allocation as well as the unique leximin allocation.
This means that the price of MEW, the price of leximin, and the strong price of leximin are all $\Omega(n)$.

It remains to show that the price of MEW and the strong price of leximin are $O(n)$.
For leximin, this follows from Lemma~\ref{lem:MNW-leximin-welfare} and the fact that the optimal social welfare is at most $n$.
We claim that for any instance, there exists a MEW allocation with social welfare at least $1$. To prove this claim, we apply the same argument as in Lemma~\ref{lem:MNW-leximin-welfare}, but starting with a MEW allocation with maximum social welfare. An improvement in the multiset of utilities as described in the argument does not decrease the egalitarian welfare and strictly increases the social welfare, which gives us the desired contradiction.
\end{proof}

Surprisingly, MEW allocations can be arbitrarily bad when there are at least three agents.

\begin{theorem}
\label{thm:MEW-PoA}
For $n>2$, the strong price of MEW is infinite.
\end{theorem}

\begin{proof}
Let $m=n$, and assume that the utilities are as follows:
\begin{itemize}
	\item $u_1(1)=1$ and $u_1(j)=0$ otherwise.
	\item For $i=2,\dots,n$: $u_i(i-1)=1$ and $u_i(j)=0$ otherwise.
\end{itemize}

Observe that in any allocation, some agent does not get a desired good. This means that every allocation has egalitarian welfare 0, and all allocations are MEW. Now, there exists an allocation with social welfare 0, for example the allocation that assigns good $i+1$ to agent $i$ for $i=1,\dots,n-1$, and assigns good 1 to agent $n$. Since there also exists an allocation with positive social welfare, the strong price of MEW is infinite.
\end{proof}

We now turn to the case of two agents. For MNW, we establish almost tight bounds on both prices of fairness.

\begin{theorem}
\label{thm:MNW-PoS-PoA-2}
For $n=2$, the price of MNW and the strong price of MNW are at least $27/23\approx 1.174$ and at most $5/4=1.25$.
\end{theorem}

\begin{proof}
It suffices to show that the price of MNW is at least $27/23$ and the strong price of MNW is at most $5/4$.

\medskip

\noindent\emph{Lower bound}: Let $m=3$ and $0<\epsilon<1/7$, and assume that the utilities are as follows:
\begin{itemize}
    \item $u_1(1)=2/3$, $u_1(2)=1/3$, $u_1(3)=0$.
    \item $u_2(1)=4/7-\epsilon$, $u_2(2)=1/7+\epsilon$, $u_2(3)=2/7$.
\end{itemize}

The optimal social welfare is $9/7$, obtained by assigning the first two goods to the first agent and the last good to the second agent. On the other hand, one can check that the maximum Nash welfare is $2/7+2\epsilon/3$, obtained (uniquely) by assigning the first good to the first agent and the last two goods to the second agent. This allocation yields social welfare $23/21+\epsilon$. Taking $\epsilon\rightarrow 0$, we find that the price of MNW is at least $27/23$.

\medskip

\noindent\emph{Upper bound}: Consider an arbitrary instance. Suppose that the optimal social welfare is $x$. If $x \leq 5/4$, then Lemma~\ref{lem:MNW-leximin-welfare} immediately implies that the price of MNW of this instance is at most $5/4$.

We now focus on the case where $x \geq 5/4$. Let us assume further that, in an optimal allocation, the first agent has utility $x_1$ and the second has utility $x_2$, where $x_1 \geq x_2$ and $x_1 + x_2 = x$. Since $x_1 \leq 1$, we have $x_1/x_2 \leq 1/(x - 1) \leq 4$.

Next, consider any MNW allocation. Suppose that in this allocation the first agent has utility $y_1$ and the second has utility $y_2$. Since the Nash welfare of this allocation must be at least that of the optimal allocation, we have $y_1y_2 \geq x_1x_2$. As a result, the social welfare of this allocation is $y_1+y_2\geq 2\sqrt{y_1y_2} \geq 2\sqrt{x_1x_2}$, where the first inequality follows from $(\sqrt{y_1}-\sqrt{y_2})^2\geq 0$. Thus, the price of MNW of this instance is at most
\begin{align*}
\frac{x_1 + x_2}{2\sqrt{x_1x_2}} &= 1 + \frac{1}{2} \cdot \left(\sqrt[4]{\frac{x_1}{x_2}} - \sqrt[4]{\frac{x_2}{x_1}}\right)^2 \\
&\leq 1 + \frac{1}{2} \cdot \left(\sqrt[4]{4} - \sqrt[4]{\frac{1}{4}}\right)^2 = 5/4,
\end{align*}
where the inequality follows from $1 \leq x_1/x_2 \leq 4$.
\end{proof}

Finally, we derive the exact bound for MEW and leximin with two agents. Note that since all leximin allocations are MEW, Theorem~\ref{thm:MEW-PoS-PoA-2} immediately implies Theorem~\ref{thm:leximin-PoS-PoA-2}.

\begin{theorem}
\label{thm:MEW-PoS-PoA-2}
For $n=2$, the price of MEW and the strong price of MEW are $3/2$.
\end{theorem}

\begin{proof}
It suffices to show that the price of MEW is at least $3/2$ and the strong price of MEW is at most $3/2$.

\medskip

\noindent\emph{Lower bound}: Let $m=3$ and $0<\epsilon<1/2$, and assume that the utilities are as follows:
\begin{itemize}
    \item $u_1(1)=1/2$, $u_1(2)=1/2-\epsilon$, $u_1(3)=\epsilon$.
    \item $u_2(1)=1/2$, $u_2(2)=\epsilon$, $u_2(3)=1/2-\epsilon$.
\end{itemize}

The optimal social welfare is $3/2-2\epsilon$, obtained by assigning the first two goods to the first agent and the last good to the second agent. On the other hand, the maximum egalitarian welfare is $1/2$, which can be obtained only by assigning the first good to one agent and the remaining two goods to the other agent. This allocation has social welfare $1$. Taking $\epsilon\rightarrow 0$, we find that the price of MEW is at least $3/2$.

\medskip

\noindent\emph{Upper bound}: Consider an arbitrary instance, and denote by $x$ the maximum egalitarian welfare. The optimal social welfare is at most $1+x$, and the social welfare of any MEW allocation is at least $2x$. Consider any MEW allocation, and suppose that agent 1 receives utility $x$ and agent 2 receives utility $y\geq x$. In the allocation where the bundles of the two agents are swapped, the utilities are $1-x$ and $1-y\leq 1-x$. Since $x$ is the maximum egalitarian welfare, we have $x\geq 1-y$, or $x+y\geq 1$. This means that the social welfare of the original allocation is at least $1$, so the social welfare of any MEW allocation is at least $\max\{2x,1\}$.

The strong price of MEW is therefore at most
$\frac{1+x}{\max\{2x,1\}}$. If $x\leq 1/2$, this quantity is at most $\frac{1+x}{1}\leq \frac{3}{2}$. On the other hand, if $x>1/2$, this quantity is at most $\frac{1+x}{2x}=\frac{1}{2x}+\frac{1}{2}<\frac{3}{2}$. The conclusion follows.
\end{proof}

\begin{theorem}
\label{thm:leximin-PoS-PoA-2}
For $n=2$, the price of leximin and the strong price of leximin are $3/2$.
\end{theorem}

\section{Pareto Optimality}

In this section, we consider Pareto optimality.
Since any allocation that maximizes social welfare is necessarily Pareto optimal, the price of Pareto optimality is trivially $1$. By establishing a tight lower bound on the welfare of a Pareto optimal allocation, we show that the strong price of Pareto optimality is quadratic. Our result indicates that while Pareto optimality is sometimes referred to as `efficiency', it does not necessarily fare well if efficiency is measured in terms of social welfare.

\begin{lemma}
\label{lem:PO-welfare}
For any instance, every Pareto optimal allocation has social welfare at least $1/n$, and this bound is tight.
\end{lemma}

\begin{proof}
To establish the bound, it suffices to show that in any Pareto optimal allocation, some agent receives utility at least $1/n$. Suppose that this is not the case. Since the utility of each agent for the entire set of goods is $1$, every agent envies at least one other agent. This implies that the \emph{envy graph}, which has the $n$ agents as its vertices and in which there is a directed edge from one agent to another if the former agent envies the latter, contains a directed cycle. By giving agent $j$'s bundle to agent $i$ for every edge $i\rightarrow j$ in the cycle, we obtain a Pareto improvement, a contradiction.

The tightness of the bound follows from the instance in Theorem~\ref{thm:PO-PoA}.
\end{proof}

\begin{theorem}
\label{thm:PO-PoA}
The strong price of Pareto optimality is $\Theta(n^2)$.
\end{theorem}

\begin{proof}
\emph{Upper bound}: Consider an arbitrary instance. Since every agent receives utility at most 1, the optimal social welfare is at most $n$. On the other hand, by Lemma~\ref{lem:PO-welfare}, every Pareto optimal allocation has social welfare at least $1/n$. The conclusion follows.

\medskip

\noindent\emph{Lower bound}: Assume that $n\geq 2$. Let $m=n$, $0<\epsilon<1/n$, and assume that the utilities are as follows:
\begin{itemize}
    \item $u_1(1)=\frac{1}{n}+\epsilon$ and $u_1(j)=\frac{1}{n}-\frac{\epsilon}{n-1}$ otherwise.
    \item For $i=2,\dots,n$: $u_i(i-1)=1-\epsilon$, $u_i(i)=\epsilon$, and $u_i(j)=0$ otherwise.
\end{itemize}

Consider the allocation that assigns good $i-1$ to agent $i$ for $i=2,\dots,n$, and good $n$ to agent 1. The welfare of this allocation is $(n-1)(1-\epsilon)+\left(\frac{1}{n}-\frac{\epsilon}{n-1}\right) = n-1+\frac{1}{n}-\left(n-1+\frac{1}{n-1}\right)\epsilon$. On the other hand, the allocation that assigns good $i$ to agent $i$ for $i=1,\dots,n$ is Pareto optimal. This is because in any Pareto improvement, agent 1 must receive good 1, and it follows that agent $i$ must receive good $i$ for every $i$. The social welfare of this allocation is $\left(\frac{1}{n}+\epsilon\right)+(n-1)\epsilon=\frac{1}{n}+n\epsilon$. Taking $\epsilon\rightarrow 0$ yields the desired result.
\end{proof}

We also show an exact bound for the case of two agents.

\begin{theorem}
\label{thm:PO-PoA-2}
For $n=2$, the strong price of Pareto optimality is $3$.
\end{theorem}

\begin{proof}
The instance in Theorem~\ref{thm:PO-PoA} shows that the strong price of Pareto optimality is at least $3$. To show that this is tight, consider an arbitrary instance and an optimal allocation in this instance. Assume that the two agents receive utility $x$ and $y$ in this allocation, where $x\geq y$. In any Pareto optimal allocation, at least one agent must receive utility at least $y$; otherwise the optimal allocation is a Pareto improvement. In combination with Lemma~\ref{lem:PO-welfare}, this implies that the social welfare of every Pareto optimal allocation is at least $\max\{y,1/2\}$.

The strong price of Pareto optimality is therefore at most $\frac{x+y}{\max\{y,1/2\}}\leq\frac{1+y}{\max\{y,1/2\}}$. If $y\leq 1/2$, this quantity is at most $2(1+y)\leq 3$. On the other hand, if $y>1/2$, this quantity is at most $\frac{1+y}{y}=\frac{1}{y}+1<3$. The conclusion follows.
\end{proof}

\section{Conclusion and Future Work}

In this paper, we study the price of fairness for indivisible goods using several fairness notions that can always be satisfied. For most cases, we exhibit tight or asymptotically tight bounds on the worst-case efficiency loss that can occur due to fairness constraints.
Interestingly, both the round-robin and MNW allocations, which are EF1, can have social welfare a linear factor away from the optimum---however, round-robin performs significantly better than this worst-case bound as long as the number of goods is not huge compared to the number of agents.
The linear bound that we obtain for MNW stands in contrast to \citet{BertsimasFaTr11}'s result in the divisible goods setting, where the price of MNW is $\Theta(\sqrt{n})$.
%In future research, it would be useful to close the gaps that remain after this work, the most intriguing of which is the EF1 gap between $\Omega(\sqrt{n})$ and $O(n)$. 

A potential direction for future work is to perform similar analyses but using egalitarian welfare instead of utilitarian welfare as the benchmark. This has been done, for example, in the context of contiguous allocations \citep{AumannDo15,Suksompong19}. 
One could also study the price of fairness for the \emph{chore division} problem, where chores refer to items that yield negative utility for the agents. 
Indeed, almost all of the notions that we consider in the goods setting have direct analogs in the chore setting, and it would be interesting to see whether the corresponding bounds in the two settings turn out to be similar as well.

\section*{Acknowledgments} 

This work is partially supported by NSF Award CCF-1815434, by the European Research Council (ERC) under grant number 639945
(ACCORD), and by an NUS Start-up Grant.
We are grateful to the reviewers of IJCAI 2019 and Theory of Computing Systems for many helpful comments, and to Ioannis Caragiannis for pointing us to the price of envy-freeness result by \citet{BertsimasFaTr11}.

% Bibliography
\bibliographystyle{plainnat}
\bibliography{ijcai19}

\end{document}